\documentclass[conference]{IEEEtran}
\IEEEoverridecommandlockouts
\usepackage{cite}
\usepackage{amsmath,amssymb,amsfonts}

\usepackage{graphicx}
\usepackage{textcomp}
\usepackage{xcolor}
\usepackage{dsfont}
\usepackage[center]{caption}
\usepackage[english]{babel}
\usepackage{color,soul}
\usepackage{amsfonts}
\usepackage{amsmath} 
\usepackage{accents}
\usepackage{amssymb}
\usepackage{cuted}
\usepackage{amsthm}
\usepackage[thinc]{esdiff}
\usepackage{enumitem}
\usepackage{amsmath, amssymb}
\usepackage{algorithm}
\usepackage{algorithmic}
\usepackage[normalem]{ulem}  
\usepackage{multibib}
\newcites{app}{References (Appendix)} 
\usepackage{amsmath,amsfonts}
\usepackage{mathtools}

\newtheorem{Proposition}{Proposition}

\newtheorem{remark}{Remark}
\usepackage{marginnote}
\usepackage{eso-pic}

 \usepackage{xcolor}
\def\BibTeX{{\rm B\kern-.05em{\sc i\kern-.025em b}\kern-.08em
    T\kern-.1667em\lower.7ex\hbox{E}\kern-.125emX}}
\definecolor{purp}{RGB}{128,0,128}

\begin{document}

\title{Co-Evolving Zero-Day Jamming: Adaptive Attack Synthesis and Graph Attention-Based Online Detection}  

\author{
Ghilas Aissou\IEEEauthorrefmark{1},
R\'{e}mi A. Chou\IEEEauthorrefmark{2},
and Taejoon Kim\IEEEauthorrefmark{1}%
\thanks{\IEEEauthorrefmark{1}Arizona State University, School of Electrical, Computer and Energy Engineering, Tempe, AZ 85287, USA.
Email:\{gaissou,  taejoonkim\}@asu.edu}%
\thanks{\IEEEauthorrefmark{2}The University of Texas at Arlington, Department of Computer Science and Engineering, Arlington, TX 76010, USA.
Email: remi.chou@uta.edu}
\thanks{This work was supported in part by NSF under grants CNS2451268, CNS2514415, ONR under grant N000142112472, and the NSF and OUSD – Research and Engineering, Grant ITE2515378, as part of the NSF Convergence Accelerator Track G.}}

\AddToShipoutPictureFG*{%
    \AtPageLowerLeft{%
        \raisebox{0.35in}{%
            \hspace{0.65in}%
            \footnotesize
            Accepted for publication in the 2026 IEEE Global Communications Conference (GLOBECOM).
        }%
    }%
}

\markboth{}%
{Shell \MakeLowercase{\textit{et al.}}: A Sample Article Using IEEEtran.cls for IEEE Journals}


\maketitle
\thispagestyle{empty}
\pagestyle{empty}

\begin{abstract}
Effective evaluation of zero-day jamming detectors requires robust adversarial models. 
However, existing attack models 
often assume prior knowledge of the target receiver,
limiting their utility as evaluation benchmarks. 
On the detection side, existing detectors fail to capture the global temporal-spectral structure of jamming behavior and cannot differentiate zero-day strategies as they emerge. 
This paper addresses these limitations through a two-pronged framework.  
First, an online detection framework is introduced that combines a graph
attention network (GAT) for temporal-spectral representation
learning with Dirichlet process (DP)-means clustering. 
This framework jointly classifies known and discovers zero-day strategies within a unified learning objective. Second,
 an inference-driven reinforcement learning (RL)
jammer is proposed as an adversarial benchmark. The jammer treats the target receiver as a black-box, infers the detector state via hypothesis testing, and optimizes  the trade-off between attack impact and stealth. 
Simulation results show that the proposed RL jammer outperforms benchmarks, achieving $33\%$ higher attack efficacy and $67\%$ higher stealth. 
The proposed detection framework against the proposed RL jammer is shown to achieve $\!20\%\!$ higher detection accuracy than the benchmarks.
\end{abstract}

\vspace{-0.4cm}
\section{Introduction}
\label{sec:intro}
While learning-based jamming detectors have demonstrated 
robust detection performance against known jammers 
\cite{pirayesh2021survey,Defeating_Super_Reactive_Jammers}, they remain susceptible  to  zero-day attacks characterized by previously unseen  signatures \cite{zhou2024}. 
To rigorously evaluate detector resilience against such threats, reinforcement learning (RL)-based jammers \cite{Jamming_Bandits_A_Novel, schutz2024linear, 11075817} can be used to autonomously develop zero-day attack behaviors via online learning. 
However, their utility as evaluation benchmarks is often hindered by two limitations. First, they largely rely on a priori knowledge of the target receiver \cite{Jamming_Bandits_A_Novel}. Second, they often overlook stealth requirements, even though maintaining a low probability of detection (LPD) is necessary to prevent countermeasures from neutralizing their attacks.

On the detection side, anomaly-based methods \cite{pirayesh2021survey, Anomaly} can be used to detect zero-day attacks, but they are limited to binary benign-or-attack outputs. Similarly, existing frameworks such as open-set recognition method 
\cite{zhang2024toward} can label zero-day attacks as an unknown class, but they lack the granularity required to distinguish between diverse zero-day attacks. 
This observation readily suggests that  the challenge must be formulated as an \emph{online discovery problem}, where unlabeled samples arrive~{sequentially} as new zero-day jamming behaviors emerge and evolve over time. 
Although online k-means \cite{cohen2021online} supports sequential online updates, it does not jointly classify known and zero-day strategies. 
Existing novel class discovery methods \cite{UNO2021, zhou2024} address joint classification and clustering but operate in offline batch settings and  require a priori knowledge of the the novel class cardinality. 

In this work, we consider a wireless communication link in which a transmitter (Alice) must secure robust communications with a receiver (Bob) in the 
presence of a jammer. 
Bob is equipped with a jamming detector designed to trigger frequency-hopping countermeasure upon identifying an attack. 
Within this framework, our primary contributions addressing the latter limitations are as follows: 
\begin{itemize}
\setlength{\itemindent}{-0.0cm}
\setlength{\leftskip}{-0.4cm}
\item We propose an online detection framework that jointly classifies known jamming strategies and discovers zero-day strategies as they emerge. The main challenge is to expand the label space upon zero-day discovery while maintaining classification stability, without the need for offline retraining or a priori knowledge of the zero-day class cardinality, which is in contrast to the methods in \cite{zhou2024, UNO2021}. This is addressed by a unified online framework that combines a graph attention network (GAT) encoder for temporal-spectral representation learning with Dirichlet process (DP)-means clustering module, enabling incremental class expansion without degrading classification of known jamming strategies.
\item To effectively evaluate the proposed detector, we introduce an inference-driven RL jammer that treats the receiver (Bob) as a \emph{black-box}. This  overcomes the reliance on receiver-side knowledge common in  prior RL-based models  \cite{Jamming_Bandits_A_Novel, schutz2024linear}. 
The jammer infers the hidden detector state at the receiver via cyclostationary feature extraction and generalized likelihood ratio test (GLRT)-based hypothesis testing ~\cite{317857}. Using the inferred state the jammer trains an RL policy to maximize interference efficiency while maintaining LPD (stealth).
\end{itemize}

Simulation results show that the proposed  RL jammer outperforms the benchmark method in \cite{schutz2024linear}, increasing attack efficiency by $33\%$ and stealth by $67\%$ while maintaining a block-box perspective. When evaluated against the proposed jammer the proposed detector  achieves 20$\%$ higher detection probability compared to the scheme in~\cite{zhou2024}. 

\emph{Notation:}
Bold lowercase letters denote column vectors $\boldsymbol{x}
\!=\![x_1,\ldots,x_N]^\top \!\in \mathbb{C}^{N\times 1}$, and bold uppercase letters denote matrices 
$\boldsymbol{X} \!\in\! \mathbb{C}^{M\times N}$, where $\mathbb{C}$ denote the sets of complex numbers. 
Calligraphic letters $\mathcal{X}$ denote sets, with $|\mathcal{X}|$ their cardinality.   
$\mathbb{E}[\cdot]$ 
denotes expectation operator.  

\section{System Model and Problem Formulation}
\label{sec:system_model}
We consider a wireless communication framework, illustrated in 
Fig.~\ref{fig:system_model}, where Alice communicates with Bob over 
$L$ orthogonal frequency channels in 
the presence of a jamming adversary. We adopt a discrete-time complex baseband representation. At time slot $t$, Alice selects 
channel $m\in\{1,\dots,L\}$ and transmits $s_t[n,m]\in\mathbb{C}$, 
$n=0,\dots,N-1$, where $N$ is the number of time samples per slot. The 
jammer transmits $u_t[n,m^{(J)}]\in\mathbb{C}$ on channel 
$m^{(J)}\in\{1,\dots,L\}$ according to a strategy that determines 
the waveform, channel selection, and power. The signal received at 
Bob is
\begin{equation}
r_t[n,m] \!=\! h_{\!AB}s_t[n,m] + \mathbb{I}[m\!=\!m^{\!(J)\!}]\,
h_{\!JB}\,u_t[n,m^{\!(J)\!}] + w_t[n,m],
\label{eq:received_signal}
\end{equation}
where $h_{AB},h_{JB}\in\mathbb{C}$ are the channel gains 
for the Alice-to-Bob and jammer-to-Bob links on channel $m$, respectively,
$\mathbb{I}[\cdot]$ is the indicator function that equals $1$ if the jammer is active on channel $m$ and $0$ when the jammer is not active on channel $m$, and $w_t[n,m] \sim  \mathcal{CN}(0, \sigma_w^2) $ is the additive zero-mean circularly symmetric complex Gaussian
noise with variance $\sigma_w^2$, independent 
across channels. Bob demodulates and decodes $r_t[n, m]$ to retrieve the transmitted frame, evaluating a cyclic redundancy check (CRC) to generate a feedback bit $A_t \in \{0,1\}$, where $A_t = 1$ denotes an acknowledgment (ACK) (CRC passed) and $A_t = 0$ a negative-acknowledgment (NACK) (CRC failed).

The jammer operates in full-duplex mode with ideal self-interference 
cancellation, sensing channel $m^{(J)}$ while transmitting. The 
jammer observation is
\begin{equation}
r_t^{(J)}[n,m^{(J)}] = h_{AJ}\,s_t[n,m]\,
\mathbb{I}[m\!=\!m^{(J)}] + w_t[n,m^{(J)}],
\label{eq:jammer_sensing}
\end{equation}
where $h_{AJ}\in\mathbb{C}$ is the Alice-to-jammer channel 
gain and $w_t[n,m^{(J)}]\sim\mathcal{CN}(0,\sigma_J^2)$. The jammer is assumed to overhear the binary feedback signal $A_t$ transmitted over a dedicated control channel, as shown in Fig.~\ref{fig:system_model}.
\begin{figure}[h!]
\vspace{-0.1cm}
\centering
\includegraphics[width=0.4\textwidth]{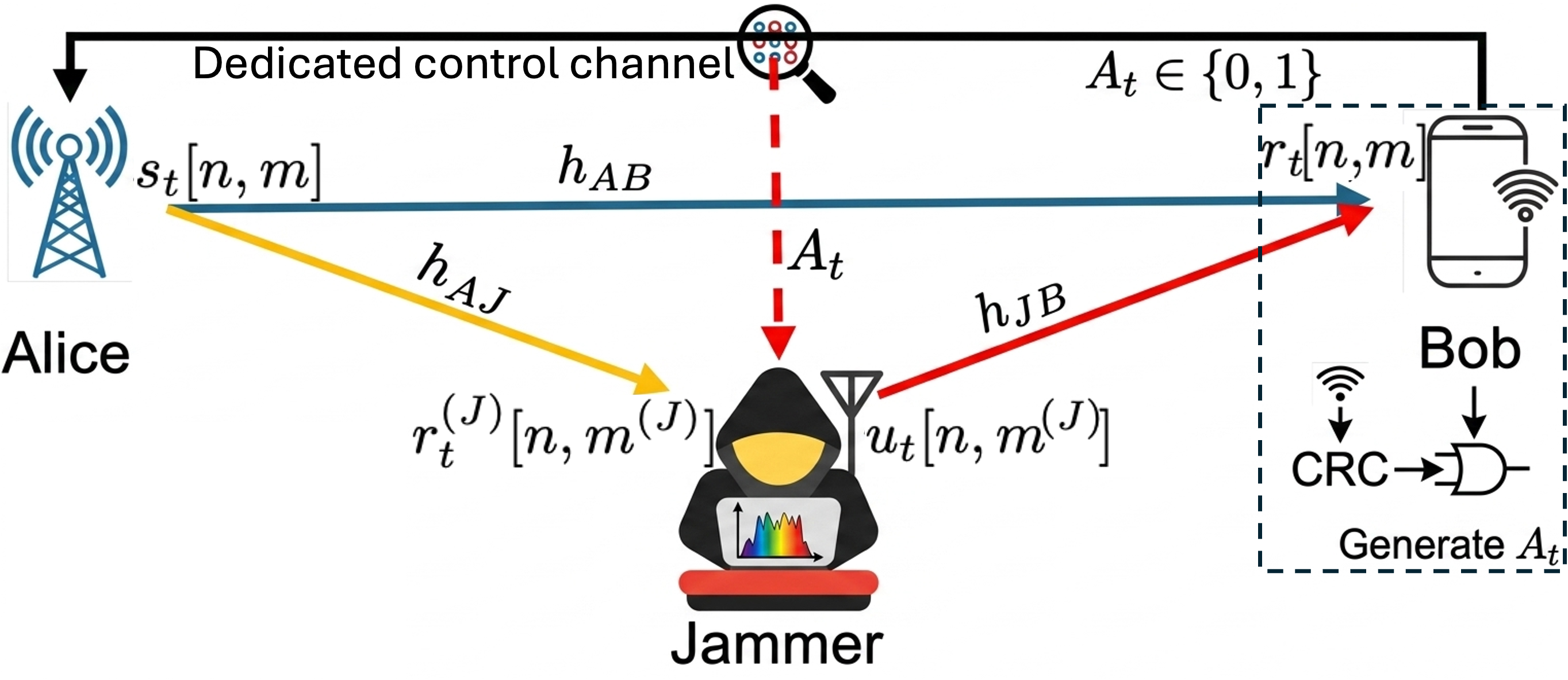}
\captionsetup{width=0.5\textwidth, justification=raggedright, 
singlelinecheck=false}
\caption{\!Wireless \!communication \!system \!under \!jamming \!attack.}
\label{fig:system_model}
\vspace{-0.2cm}
\end{figure}
\vspace{-0.3cm}

\subsection*{Jamming Detector and Problem Statement}
\label{Jamming_Detector_and_Problem_Statement}
Bob constructs a time-frequency frame $\boldsymbol{X}_{\!k}\!\in\!\mathbb{R}^{L\!\times\! T}$~over $T$ consecutive time slots, indexed by $\tau \in \{1, \dots, T\}$ corresponding to the slot index $t = (k-1)T + \tau$. The entry $[\boldsymbol{X}_k]_{m,\tau}$ denotes the average received energy on channel $m$ during local slot $\tau$, defined as $[\boldsymbol{X}_k]_{m,\tau} = \frac{1}{N}\sum_{n=0}^{N-1} \left| r_t[n,m] \right|^2.$
The detector at Bob is trained on a labeled dataset $\mathcal{D}=\{(\boldsymbol{X}_i, y_i)\}_{i=1}^{N_s}$, where $N_s$ is the number of training frames, $y_i = 1$ denotes benign transmission, and $y_i \in \{2, \dots, K\}$ denotes one of $K-1$ known jamming strategies. The detector output is a 
binary decision $d_k\in\{0,1\}$, where $d_k=1$ indicates jamming 
detection and $d_k=0$ indicates no jamming detection. Upon detection ($d_k=1$), Alice and Bob activate a frequency-hopping 
countermeasure in frame $k+1$. While we adopt this concrete setting for our analysis, the specific architecture of the detector and countermeasure remains modular and secondary to the proposed zero-day framework.

The detector is trained on $K$ known classes and is not designed to 
handle zero-day jamming. During deployment, Bob receives an unlabeled 
stream $\{\boldsymbol{X}_k\}_{k=1}^{\infty}$, where each frame may 
belong to a known or a zero-day strategy. The 
task is formulated as an online discovery problem. We want to produce a decision 
$c_k\in\{1,\dots,K\!+\!M_k\}$ for each incoming $\boldsymbol{X}_k$, where 
indices $1,\dots,K$ correspond to known classes and indices 
$K\!+\!1,\dots,K\!+\!M_k$ correspond to discovered zero-day strategies, 
with $M_k$ the number of zero-day strategies discovered up to frame $k$. The challenge is that the detector must expand its label space upon zero-day discovery, with no prior knowledge of how many strategies $M_k$ will emerge. At the same time, the detector must remain stable on known classes to prevent error propagation during subsequent online updates. We next present the online detection framework addressing these challenges.

\section{Online Zero-Day Jamming Clustering}
\label{sec:online_detection}

Fig.~\ref{fig:Detection_framework_globecom} illustrates the proposed detection framework, whose components are detailed in the following subsections. A GAT encoder maps the frame $\boldsymbol{X}_k$ to a latent embedding, which a dynamic classifier scores against the current label space to produce $c_k$. A DP-means module determines whether the embedding belongs to an existing class or creates a new zero-day class, while a multi-view consistency loss keeps the encoder discriminative on the unlabeled deployment stream. We describe the offline pre-training first, then the online~detection.

\begin{figure*}[h!]
    \centering    \includegraphics[width=0.95\textwidth]{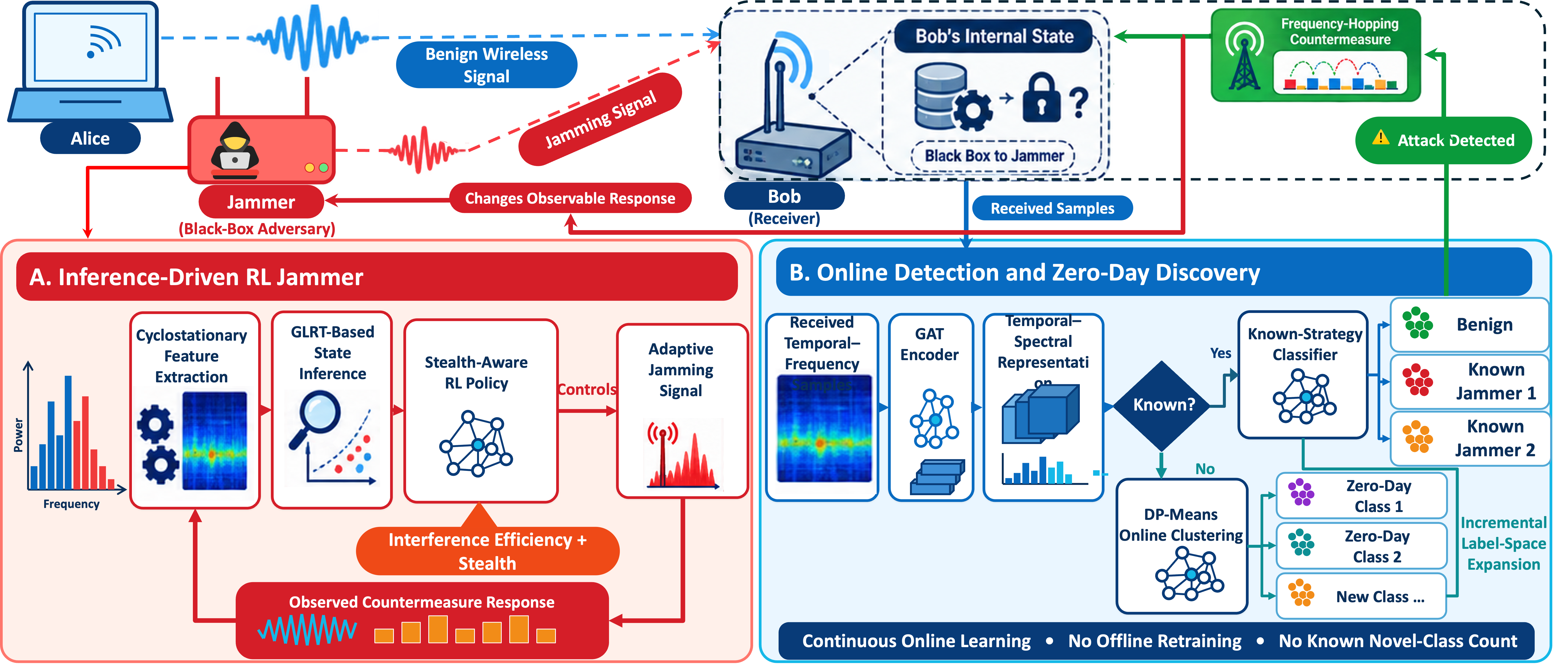}
    \captionsetup{width=\textwidth, justification=raggedright, singlelinecheck=false}
    \caption{Proposed online zero-day jamming detection and inference-driven RL-based jamming synthesis framework.}
\label{fig:Detection_framework_globecom}
    \vspace{-0.48cm}
\end{figure*}

\subsection{GAT Encoder and Offline Pre-Training}
\label{subsec:encoder_pretrain}

We describe the GAT encoder architecture and the pre-training 
procedure that initializes the encoder, classifier, and DP-means 
radius before deployment.

\subsubsection{Graph Construction and Attention}
Jamming strategies induce structured global dependencies across frequency channels and time. Classical encoders, such as convolutional encoders, are not suited to modeling such global structure due to their limited receptive fields~\cite{9672167}. To model these non-local 
dependencies, we represent each frame  $\boldsymbol{X}_k\!\in\!\mathbb{R}^{L\!\times\! T}$ as a fully connected graph over $LT$ nodes, where node $(m,\tau)$ corresponds to entry $[\boldsymbol{X}_k]_{m,\tau}$. Each node $(m,\tau)$ carries the feature vector $\boldsymbol{x}_{k,m,\tau}\!=\!\phi_{\theta_\phi}([\boldsymbol{X}_k]_{m,\tau})\in\mathbb{R}^{d_0\times 1}$, with $\phi_{\theta_\phi}(\cdot)$ a learnable node-wise map, with parameter $\theta_\phi$. 
For notational convenience, nodes are indexed by $i\!=\!(m\!-\!1)T\!+\!\tau\!\in\!\{1,\dots,LT\}$, so $\boldsymbol{x}_{k,i}\triangleq\boldsymbol{x}_{k,m,\tau}$. 

The feature $\boldsymbol{x}_{k,i}$ is projected into query, key, and value vectors $\boldsymbol{q}_{k,i}\!\!=\!\!\boldsymbol{W}_Q^\top\boldsymbol{x}_{k,i}$, $\boldsymbol{k}_{k,i}\!\!=\!\!\boldsymbol{W}_K^\top\boldsymbol{x}_{k,i}$, $\boldsymbol{v}_{k,i}\!\!=\!\!\boldsymbol{W}_V^\top\boldsymbol{x}_{k,i}$, respectively, where $\boldsymbol{W}_Q,\boldsymbol{W}_K,\boldsymbol{W}_V\!\in\!\mathbb{R}^{d_0\!\times\! d}$ are projection matrices and $d$ is the embedding dimension \cite{vaswani2017attention}. The attention weight~from node $\!j\!$ to node $i$ is $e_{k,ij} \!=\! \frac{\exp(\boldsymbol{q}_{k,i}^\top\boldsymbol{k}_{k,j}/\sqrt{d})}{\sum_{\ell=1}^{LT}\exp(\boldsymbol{q}_{k,i}^\top\boldsymbol{k}_{k,\ell}/\sqrt{d})}.$
Each~node aggregates value vectors weighted by attention, $\boldsymbol{h}_{k,i}\!=\!\sum_{j=1}^{LT}e_{k,ij}\boldsymbol{v}_{k,j} \!\in\! \mathbb{R}^{d\times 1}$, and mean pooling yields the frame embedding $\boldsymbol{z}_k \!=\! \frac{1}{LT}\sum_{i=1}^{LT}\boldsymbol{h}_{k,i}\in\mathbb{R}^{d\times 1}$ \cite{vaswani2017attention}.
The GAT encoder parameters $\theta_G\!\triangleq\!(\theta_\phi,\!\boldsymbol{W}_Q,\!\boldsymbol{W}_K,\!\boldsymbol{W}_V)$ define the~mapping~$\boldsymbol{z}_k\!=\!f_{\theta_G}(\!\boldsymbol{X}_k\!)$. 

\subsubsection{Offline Pre-Training}
As shown in Fig.~\ref{fig:Detection_framework_globecom} (in blue), the encoder $f_{\theta_G}$ and initial classifier 
$\boldsymbol{W}_{\mathrm{cls},0}=[\boldsymbol{w}_1,\dots,\boldsymbol{w}_K]
\in\mathbb{R}^{d\times K}$, where $\boldsymbol{w}_c\in\mathbb{R}^{d\times 1}$ 
is the weight vector for class $c \!\in\! \{1,\cdots,K\}$, are trained  on the dataset $\mathcal{D}$ 
by minimizing the supervised loss $\mathcal{L}^\ell = -\frac{1}{N_s}\sum_{(\boldsymbol{X}_i,y_i)\in\mathcal{D}}
\log p(y_i\mid\boldsymbol{X}_i).$
Here, $p(c\mid\boldsymbol{X}_i) = 
{\exp(\boldsymbol{w}_c^\top f_{\theta_G}(\boldsymbol{X}_i))}/
{\sum_{c'=1}^{K}\exp(\boldsymbol{w}_{c'}^\top f_{\theta_G}(\boldsymbol{X}_i))}$
is the softmax posterior. At the end of training, the centroid 
of each known class $c\in\{1,\dots,K\}$ is
\vspace{-0.1cm}
\begin{equation}
\boldsymbol{\mu}_c = \frac{1}{|\mathcal{I}_c|}\sum_{i\in\mathcal{I}_c}f_{\theta_G}(\boldsymbol{X}_i) \in \mathbb{R}^{d\times 1}, 
\label{eq:known_centroids}
\vspace{-0.1cm}
\end{equation}
where $\mathcal{I}_c=\{i:y_i=c\}$. These centroids $\{\boldsymbol{\mu}_c\}_{c=1}^K$ play a dual role: they initialize the DP-means representative set at deployment and provide the reference geometry for calibrating the DP-means radius $\lambda_{\mathrm{dp},c}$. Specifically, $\lambda_{\mathrm{dp},c}$ is set on a held-out validation set $\mathcal{D}_{\mathrm{val}}$ as $\lambda_{\mathrm{dp},c} = \max_{\substack{(\boldsymbol{X}_i,y_i)\in\mathcal{D}_{\mathrm{val}},y_i=c}}\|f_{\theta_G}(\boldsymbol{X}_i)-\boldsymbol{\mu}_c\|_2, \quad c\in\{1,\dots,K\}.$
 During deployment, a frame is assigned to its nearest class only if the distance falls within that class's own radius; otherwise it is flagged as a zero-day candidate.

\subsection{Online Classification and Zero-Day Discovery}
\label{subsec:online_classification}
At deployment, the framework draws an unlabeled minibatch $\mathcal{B}_k^u  $ from the deployment stream, and processes $\boldsymbol{X}_k \in \mathcal{B}_k^u$  in three sequential steps: (i)~two augmented views are generated and embedded to produce reliable latent representations; (ii)~the DP-means module assigns each frame to an existing class or instantiates a new one; and (iii)~the dynamic classifier produces a posterior over the current label space.
 
\subsubsection{Multi-View Augmentations and Unlabeled Consistency}
Each frame $\boldsymbol{X}_k \!\in\! \mathcal{B}_k^u$ is transformed into two stochastic views 
$\boldsymbol{X}_k^{(1)}\!=\!\Psi_1(\boldsymbol{X}_k) \!\in\! \mathbb{R} ^{L\!\times\!T}$ and 
$\boldsymbol{X}_k^{(2)}\!=\!\Psi_2(\boldsymbol{X}_k)\!\in \!\mathbb{R} ^{L\!\times\!T}$, where $\Psi_1$ applies a random time shift and 
$\Psi_2$ applies a Gaussian noise perturbation. The perturbation magnitudes are calibrated on $\mathcal{D}$ to preserve the class label, in the sense that 
$\arg\max_c p(c \mid \Psi_v(\boldsymbol{X}_k)) = \arg\max_c p(c \mid \boldsymbol{X}_k)$ holds for $v\in\{1,2\}$. The shared encoder produces embeddings $ {\boldsymbol{z}}_k^{(v)}=f_{\theta_G}( {\boldsymbol{X}}_k^{(v)})$ and classifier outputs posterior vectors $ {\boldsymbol{p}}_k^{(v)}=[p_k(1\mid {\boldsymbol{X}}_k^{(v)}),\dots,p_k(K+M_{k-1}\mid {\boldsymbol{X}}_k^{(v)})]^\top\in[0,1]^{(K\!+\!M_{k-1})\!\times \! 1}$. Similar to ~\cite{UNO2021}, each posterior is converted by the Sinkhorn--Knopp algorithm into a soft pseudo-label $ {\boldsymbol{b}}_k^{(v)}\in[0,1]^{(K\!+\!M_{k-1})\! \times \! 1}$ enforcing a uniform class marginal. The pseudo-label of one view serves as the target for the other, yielding a consistency loss
\vspace{-0.1cm}
\begin{equation}
\begin{aligned}
\mathcal{L}_k^u \!=\!\!\frac{-1}{|\mathcal{B}_k^u|}\!\!\sum_{\boldsymbol{X}_u\!\in\!\mathcal{B}_k^u}\!\!\!\!\!\sum_{c=1}^{K\!+\!M_{k\!-\!1}}\!\!\!\Bigl[  {b}_{u,c}^{(\!2\!)}\log p_u(c\!\mid\!    {\boldsymbol{X}}_u^{(\!1\!)})\! +\!   {b}_{u,c}^{(\!1\!)}\log p_u(c\!\mid\!    {\boldsymbol{X}}_u^{(\!2\!)}\!)\!\Bigr].
\label{eq:unlabeled_loss}
\end{aligned}
\vspace{-0.1cm}
\end{equation}
This loss prevents the encoder from drifting on the unlabeled stream. The view embeddings $( {\boldsymbol{z}}_k^{(1)},  {\boldsymbol{z}}_k^{(2)})$ are then passed to the DP-means module.

\subsubsection{DP-Means Assignment and Class Expansion}
 
The DP-means module maintains representatives $\{\boldsymbol{\mu}_c\}_{c=1}^{K+M_{k-1}}$, initialized from~\eqref{eq:known_centroids} and extended as new classes are discovered. Each view embedding is independently assigned by comparing its distance to the nearest representative against that class's calibrated radius,
\begin{equation}
\begin{aligned}
c_k^{(v)} = \begin{cases}
c^\star, & \| {\boldsymbol{z}}_k^{(v)}-\boldsymbol{\mu}_{c^\star}\|_2 \le \lambda_{\mathrm{dp},c^\star},\\
\mathrm{NEW}, & \text{otherwise},
\end{cases}
\label{eq:view_assignment}
\end{aligned}
\end{equation}
where $c^\star = \arg\min_c\| {\boldsymbol{z}}_k^{(v)}-\boldsymbol{\mu}_c\|_2$. Consensus between both views is required before any classification is made, guarding against spurious decisions on ambiguous frames. Three outcomes follow as shown in Fig.~\ref{fig:Detection_framework_globecom}:
 
\textit{Case 1: Agreement on an existing class} ($c_k^{(1)}\!=\!c_k^{(2)}=c^\star\!\neq\!\mathrm{NEW}$): the representative is updated with the averaged embedding $\bar{\boldsymbol{z}}_k=\frac{1}{2}( {\boldsymbol{z}}_k^{(1)}+ {\boldsymbol{z}}_k^{(2)})$ via
\begin{equation}
\boldsymbol{\mu}_{c^\star} \leftarrow \boldsymbol{\mu}_{c^\star} + \eta(\bar{\boldsymbol{z}}_u - \boldsymbol{\mu}_{c^\star}),
\label{eq:centroid_update}
\end{equation}
where $\eta \in (0,1]$ is the step size, and its cluster radius is updated online as $\lambda_{\mathrm{dp},c^\star} \leftarrow \max\left(\lambda_{\mathrm{dp},c^\star}, \|\bar{\boldsymbol{z}}_k - \boldsymbol{\mu}_{c^\star}\|_2\right)$.

\textit{Case 2 -- Agreement on a new class} ($c_k^{(1)}=c_k^{(2)}=\mathrm{NEW}$): a new class $c_{\mathrm{new}}=K+M_{k-1}+1$ is created with $\boldsymbol{\mu}_{c_{\mathrm{new}}}\leftarrow\bar{\boldsymbol{z}}_k$ and $\lambda_{\mathrm{dp},c_{\mathrm{new}}}\!\leftarrow\!\| {\boldsymbol{z}}_k^{(1)}\!-\! {\boldsymbol{z}}_k^{(2)}\|_2/2$. The classifier is expanded~as
\begin{equation}
\boldsymbol{W}_{\mathrm{cls},k}\leftarrow[\boldsymbol{W}_{\mathrm{cls},k-1}\ \ \boldsymbol{\mu}_{c_{\mathrm{new}}}],
\label{eq:classifier_expansion}
\end{equation}
and $M_{k}$ is incremented to $M_{k}=M_{k-1}+1$.
 
\textit{Case 3 --- Disagreement} ($c_k^{(1)}\neq c_k^{(2)}$): the frame is discarded from the representative update and pseudo-label generation, but still contributes to $\mathcal{L}_k^u$.
 
The DP-means objective penalizes large distances between samples and their assigned cluster centers, while discouraging the creation of unnecessary new clusters:
\begin{equation}
\mathcal{J}_{\mathrm{cluster},k} = \sum_{\boldsymbol{X}_i\in\mathcal{B}_k^u}\sum_{c}r_{i,c}\|\bar{\boldsymbol{z}}_i-\boldsymbol{\mu}_c\|_2^2 + {\lambda}_{\mathrm{dp}}^2 M_{k-1},
\label{eq:cluster_objective_detector}
\end{equation}
where $r_{i,c}=1$ iff both views agree on class $c$, and ${\lambda}_{\mathrm{dp}}$ is the penalty for creating a new class, set to ${\lambda}_{\mathrm{dp}} = \min_{c} \lambda_{\mathrm{dp},c}$ so that any sample flagged as NEW in \eqref{eq:view_assignment} incurs this penalty.

At step $k$, the classifier $\boldsymbol{W}_{\mathrm{cls},k}\in\mathbb{R}^{d\!\times\!(K\!+\!M_{k\!-\!1})}$ scores $\boldsymbol{z}_k$ via
$p(y_k=c\!\mid\!\boldsymbol{X}_k) \!=\! \frac{\exp(\boldsymbol{w}_c^\top\boldsymbol{z}_k)}{\sum_{c'=1}^{K+M_{k-1}}\exp(\boldsymbol{w}_{c'}^\top\boldsymbol{z}_k)},$
and produces $c_k=\arg\max_c p(c\mid\boldsymbol{X}_k)$. Known and zero-day classes are handled uniformly within the same softmax. When Case~2 occurs, the classifier is expanded via~\eqref{eq:classifier_expansion} with the new column initialized from the discovered centroid.

\subsubsection{Unified Online Objective}
\label{subsec:detector_objective}
 
The supervised loss $\mathcal{L}_k^\ell$, consistency loss $\mathcal{L}_k^u$, and clustering objective $\mathcal{J}_{\mathrm{cluster},k}$ all operate on the same encoder and classifier, so they are not independently optimized but jointly solved through
\begin{equation}
\mathcal{L}_{\mathrm{total},k} = \mathcal{L}_k^\ell + \frac{1}{2e^{2s_u}}\mathcal{L}_k^u + \frac{1}{2e^{2s_c}}\mathcal{J}_{\mathrm{cluster},k} + s_u + s_c,
\label{eq:total_loss_detector}
\end{equation}
where $\mathcal{L}_k^\ell\!=\!{-|\mathcal{B}_k^\ell|}^{-1}\!\sum_{(\boldsymbol{X}_i,y_i)\in\mathcal{B}_k^\ell}\log p(y_i\mid\boldsymbol{X}_i)$ is evaluated on a labeled replay minibatch $\mathcal{B}_k^\ell \subset \mathcal{D}$, and $s_u,s_c\in\mathbb{R}$ are trainable log-uncertainty weights~\cite{kendall2018multi} that balance the three terms without manual tuning. At each step, $\{\theta_G,\boldsymbol{W}_{\mathrm{cls},k},s_u,s_c\}$ are updated by gradient descent on~\eqref{eq:total_loss_detector} and $\{\boldsymbol{\mu}_c\}$ by~\eqref{eq:centroid_update}, coupled within a single online objective. 

To evaluate the proposed zero-day detector, we next introduce the inference-driven RL jammer. 

\section{Inference-Driven RL Jammer}
\label{Inference-Driven RL-Based Jammer}
We model the jammer as an inference-driven adversary that learns its strategy through RL, with the objective of maximizing the outage probability at Bob while remaining~undetected.
\subsection{Jammer Observations and Detector State Estimation}
As illustrated in Fig.~\ref{fig:system_model} (shown in red), the jammer observes the ACK/NACK feedback $A_t$~\cite{Kim_paper_ACK} but does not directly observe the detector decision $d_k$ or the countermeasure activation. To bridge this information gap and inform its RL-based policy, the jammer must probabilistically infer the detector's state. We formulate this state estimation as a binary hypothesis test on the previous frame's detection outcome:
\begin{equation}
\mathcal{H}_0:\ d_{k-1}=0,
\qquad
\mathcal{H}_1:\ d_{k-1}=1,
\label{eq:hypotheses}
\end{equation}
where $\mathcal{H}_0$ denotes stealth maintained and $\mathcal{H}_1$ denotes detection followed by countermeasure activation. The key challenge is that ACK/NACK feedback reflects both channel impairments and jamming effects and therefore does not reveal whether countermeasure activation has occurred. The jammer must infer $d_{k-1}$ from its sensing observation $r^{(J)}[n, m^{(J)}]$ in~\eqref{eq:jammer_sensing}.

To sense Alice's transmission, the jammer exploits the cyclostationary structure of Alice's signal via the spectral correlation function (SCF)~\cite{317857,5072368}. The SCF maps the signal into a two-dimensional cyclic-spectral domain parameterized by spectral frequency $\omega$ and cyclic frequency $\alpha$, providing robustness to noise uncertainty and channel fading, unlike energy-based detection. Using the smoothed cyclic periodogram $\widehat{S}_{r^{(J)},t}^{\alpha}(\omega)$~\cite{317857}, the jammer identifies a discriminative coordinate $(\alpha^\star,\omega^\star)$ during an initial observation phase and computes the frame-level SCF feature
\begin{equation}
\widehat{q}_k \triangleq \frac{1}{T}\sum_{t=(k-1)T+1}^{kT}\left|\widehat{S}_{r^{(J)},t}^{\alpha^\star}(\omega^\star)\right|.
\label{eq:frame_stat}
\end{equation}
By the frequency-shift property of the SCF~\cite{5072368}, when the countermeasure forces Alice to hop from the monitored channel (center frequency $\nu_m$) to another channel ($\nu_{m'}, m' \neq m$), the theoretical SCF satisfies $S_{r^{(J)},t}^{\alpha}(\omega) = |h_{AJ}|^2 S_s^{\alpha}(\omega - \nu_{m'})$, shifting its peak away from $(\alpha^\star,\omega^\star)$ and causing $\widehat{q}_k$ to drop. 

We define the unshifted SCF magnitude $\Psi_m \triangleq |S_s^{\alpha^\star}(\omega^\star - \nu_m)|$, the occupancy ratio $\rho_k \!\triangleq \! T^{-1}\!\sum_{t=(k\!-\!1)T\!+\!1}^{kT}\!\boldsymbol{1}\!\{\nu_{m'} = \nu_m\}$, the SCF leakage residual $\epsilon \triangleq \max_{\nu_{m'} \neq \nu_m}|S_s^{\alpha^\star}(\omega^\star - \nu_{m'})|$, and assume channel-gain bounds $0 < h_{\min}^2 \leq |h_{AJ}|^2 \leq h_{\max}^2$.

Using the asymptotic normality of the SCF estimator established in~\cite{317857}, we derive the separation guarantee for our hypothesis test. For a target error probability $\beta \in (0,1)$, define the confidence parameter $\delta_{N,T}(\beta) \triangleq \sigma_S\sqrt{2\log(2/\beta)/(NT)}$, where $\sigma_S^2$ is the asymptotic variance of $\widehat{S}_{r^{(J)},t}^{\alpha^\star}(\omega^\star)$~\cite{317857}.

\begin{Proposition}
\label{prop:scf_separability}
For any $\beta \in (0,1)$, with probability at least $1-\beta$, under $\mathcal{H}_0$
\vspace{-0.125cm}
\begin{equation}
\vspace{-0.125cm}
\widehat{q}_k \ge h_{\min}^2 \Psi_m - \delta_{N,T}(\beta),
\label{eq:qhat_H0}
\end{equation}
and under $\mathcal{H}_1$,
\begin{equation}
\widehat{q}_k \le \rho_k h_{\max}^2 \Psi_m + (1-\rho_k)h_{\max}^2\epsilon + \delta_{N,T}(\beta).
\label{eq:qhat_H1}
\end{equation}
Consequently, $\widehat{q}_k$ is discriminative with probability at least $1-\beta$ whenever
\begin{equation}
\Delta_q \triangleq (h_{\min}^2 \!-\! \rho_k h_{\max}^2)\Psi_m \!-\! (1\!-\!\rho_k)h_{\max}^2\epsilon \!>\! 2\delta_{N,T}(\beta).
\label{eq:separation_condition}
\end{equation}
\end{Proposition}

\begin{proof}
Since $w^{(J)}[n,m^{(J)}]$ is second-order stationary, its SCF satisfies $S_{w^{(J)},t}^{\alpha^\star}(\omega)=0$ for $\alpha^\star \neq 0$~\cite{317857}. By independence of $s[n,m]$ and $w^{(J)}[n,m^{(J)}]$, the nonzero-cyclic-frequency SCF of the jammer's sensed signal is governed by the Alice component, yielding $S_{r^{(J)},t}^{\alpha^\star}(\omega)=|h_{AJ}|^2 S_s^{\alpha^\star}(\omega-\nu_{m'})$. Under $\mathcal{H}_0$, all slots satisfy $\nu_{m'}=\nu_m$, so $|S_{r^{(J)},t}^{\alpha^\star}(\omega^\star)| \ge h_{\min}^2 \Psi_m$. Under $\mathcal{H}_1$, slots with $\nu_{m'} \neq \nu_m$ contribute at most $h_{\max}^2\epsilon$. Averaging over $T$ slots yields
\begin{equation}
\rho_k h_{\max}^2 \Psi_m + (1-\rho_k) h_{\max}^2 \epsilon \ge q_k \ge h_{\min}^2 \Psi_m,
\label{eq:q_bounds}
\end{equation}
where $q_k \!\triangleq \! T^{-1}\!\sum_{t=(k-1)T+1}^{kT}\!|S_{r^{(J)},t}^{\alpha^\star}(\omega^\star)|$ is the theoretical counterpart of $\widehat{q}_k$. From~\cite{317857}, $\sqrt{N}(\widehat{S}_{r^{(J)},t}^{\alpha^\star}(\omega^\star) - S_{r^{(J)},t}^{\alpha^\star}(\omega^\star)) \xrightarrow{d} \mathcal{CN}(0,\sigma_S^2)$. Applying a Gaussian tail bound to the frame-level SCF feature $\widehat{q}_k$ yields $P(|\widehat{q}_k - q_k| > \delta_{N,T}(\beta)) \le \beta$. Combining with~\eqref{eq:q_bounds} yields~\eqref{eq:qhat_H0}--\eqref{eq:qhat_H1}, and~\eqref{eq:separation_condition} follows by requiring the $\mathcal{H}_0$ lower bound to exceed the $\mathcal{H}_1$ upper bound.
\end{proof}

\begin{remark}
Rearranging~\eqref{eq:separation_condition}, the minimum observation length required to guarantee discrimination at confidence $1-\beta$ is $NT > 2\sigma_S^2\log(2/\beta)/\Delta_q^2$, providing a design rule for selecting the frame parameters $N$ and $T$ as a function of the channel conditions and target reliability.
\end{remark}

 From the asymptotic normality established in~\cite{317857}, $\widehat{q}_k \sim \mathcal{N}(q_k, \sigma_S^2/(NT))$, where $q_k$ depends on the unknown channel gains through~\eqref{eq:q_bounds}. Since these gains are unknown, the jammer applies the GLRT
\begin{equation}
\log\frac{\max_{\vartheta_1}\,p(\widehat{q}_k\mid\mathcal{H}_1;\vartheta_1)}{\max_{\vartheta_0}\,p(\widehat{q}_k\mid\mathcal{H}_0;\vartheta_0)} \mathop{\gtrless}_{\mathcal{H}_0}^{\mathcal{H}_1} \eta,
\label{eq:glrt}
\end{equation}
where $\vartheta_i$ collects the unknown channel parameters under $\mathcal{H}_i$. The threshold $\eta$ is calibrated during the observation phase to satisfy the Neyman--Pearson constraint $P_{\mathrm{FA}} \triangleq \Pr(g_k = 1 \mid \mathcal{H}_0) \leq \alpha_{\mathrm{FA}}$ for a prescribed false-alarm level $\alpha_{\mathrm{FA}}$. The GLRT outcome is mapped to
\begin{equation}
g_k = \begin{cases} 1, & \text{if the GLRT selects }\mathcal{H}_1,\\ 0, & \text{otherwise,} \end{cases}
\label{eq:gt_mapping}
\end{equation}
so that $g_k$ serves as a one-frame-delayed noisy proxy for Bob's hidden detector response $d_{k-1}$, with $g_k\!=\!1$ indicating inferred detection ($d_{k\!-\!1}\!\!=\!\!1$) and $g_k\!\!=\!0$ indicating no detection ($d_{k\!-\!1}\!\!=\!\!0$).

\subsection{Jamming Decision-Making and Learning Framework}
\label{Section_RL_jamming}

At decision frame $k$, the jammer selects an action $a_k \in \mathcal{A}$ following a policy $\pi(a_k \mid o_k)$, where $o_k$ is its observable state. The action space $\mathcal{A}$ consists of all combinations of a type of jamming waveform (e.g., barrage noise, pulsed, etc.), a target frequency channel $m^{(J)}\in \{1, \dots, L\}$, and a transmit power level subject to $\frac{1}{TN}\sum_{t=(k-1)T+1}^{kT}\sum_{n=0}^{N-1}
\mathbb{E}\bigl[|u_t[n,m^{(J)}]|^2\bigr] \leq P_{\max}.$ Each action $a_k \in \mathcal{A}$ specifies the jamming signal $u[n,m^{(J)}]$ in \eqref{eq:received_signal} during frame $k$.

The observable state is $o_k \triangleq (\tilde{q}_k, g_k, \tilde{A}_{k-1})$, where $\tilde{q}_k$ and $\tilde{A}_{k-1}$ are uniform quantizations of $\widehat{q}_k$ and $\bar{A}_{k-1} \triangleq T^{-1}\sum_{t=(k-2)T+1}^{(k-1)T} A_t$ into $q_{\mathrm{lev}}$ and $A_{\mathrm{lev}}$ levels, respectively. The resulting state space is of size $|\mathcal{O}| \!=\! q_{\mathrm{lev}} \!\times\! 2\! \times\! A_{\mathrm{lev}}$, which is controllable by the quantization level for tractable Q-learning.

The reward function must be designed to capture the trade-off between communication disruption (impact) and stealth. A more aggressive jammer disrupts communication but also increases the risk of detection. Once detected, the countermeasure forces Alice to hop to a new channel, and the jammer must spend additional frames re-sensing to locate the new channel before resuming jamming. The per-frame reward is
\begin{equation}
R_k = \underbrace{(1 - \bar{A}_k)}_{\text{impact}} - \zeta\,\underbrace{g_{k+1}}_{\text{stealth penalty}},
\label{eq:reward}
\end{equation}
where $g_{k+1}$ is the GLRT-based inference of Bob's detector decision $d_k$, available with one-frame delay, and $\zeta\geq 0$ governs the trade-off. The first term rewards communication disruption; the second penalizes inferred detection. Since $g_{k+1}$ becomes available at frame $k+1$, the Q-learning update is performed once $g_{k+1}$ becomes available.  Under the adopted quantized state representation, the learning problem is treated as a finite-state MDP approximation.

The jammer maximizes the expected discounted return $\max_{\pi}\,\mathbb{E}_{\pi}[\sum_{k=1}^{\infty}\gamma^{k-1}R_k]$ with discount factor $\gamma \in (0,1)$ via the Q-learning update with learning rate $\xi \in (0,1)$.
\begin{equation}
Q(o_k, a_k) \!\!\leftarrow \! (1 \!- \!\xi)\,\!Q(o_k, a_k) \!+ \!\xi\!\left[\!R_k \!+\! \gamma\max_{a'\in \mathcal{A}} \!Q(o_{k\!+\!1},\! a')\!\right].
\label{eq:q_update}
\end{equation}

\begin{remark}
A fixed stationary policy corresponds to a single jamming strategy, regardless of the frame-to-frame variability in the actions drawn from that policy. This follows from the equivalence between behavior strategies and mixed strategies in games with perfect recall~\cite{kuhn1953extensive}.
\end{remark}

\section{Simulation and Results}
\label{sec:sim_setup}
Alice communicates with Bob over \(L=64\) orthogonal frequency channels. Each
frame contains \(T=100\) slots, and each slot contains \(N=80\) complex
baseband samples. The Alice-to-Bob and jammer-to-Bob links are modeled as
independent Rayleigh fading channels. Alice transmits with fixed normalized power, \(P_A=\mathbb{E}[|s_t[n,m]|^2]\), while the noise variance \(\sigma_w^2\) is set such that \(10\log_{10}(P_A/\sigma_w^2)=20\) dB during training. During evaluation, \(\sigma_w^2\) is adjusted such that \(10\log_{10}(P_A/\sigma_w^2)\) varies from $5$ dB to \(30\) dB to evaluate the jammer under different channel conditions.  Bob is equipped with a convolutional neural network based jamming detector trained on $K=4$ known classes as discussed in Section~\ref{sec:system_model}:  benign transmission, sequential jammer, random jammer, and a reactive jammer \cite{pirayesh2021survey}.
In this simulation, the proposed jammer action space is instantiated by eight waveform-power
combinations, formed by four jamming power levels and two jamming waveforms: wideband noise
and narrowband tone interference. To generate zero-day jamming strategies, three independently initialized policies of the proposed RL jammer are trained. In view of Remark~2, each learned stationary policy is interpreted as one jamming strategy. Three baseline jammers are considered: a LinUCB contextual bandit~\cite{schutz2024linear}, a super-reactive jammer (SRJ)~\cite{Defeating_Super_Reactive_Jammers}, and a random jammer \cite{pirayesh2021survey}. 
\begin{figure*}[h!]
    \centering    \includegraphics[width=\textwidth]{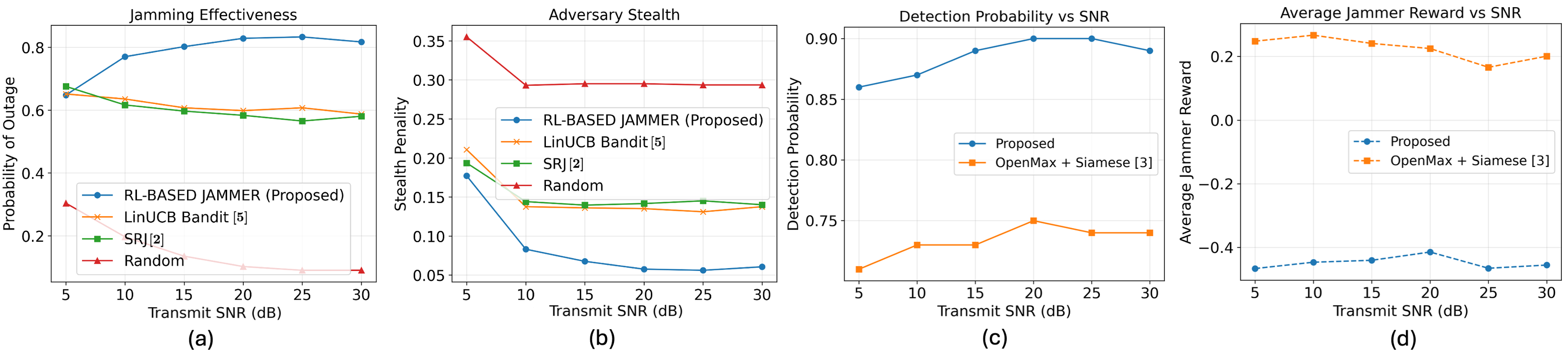}
    \captionsetup{width=\textwidth, justification=raggedright, singlelinecheck=false}
    \caption{Performance comparison across transmit SNR values: (a) jamming effectiveness, (b) adversary stealth, (c) detector detection probability, and (d) average jammer reward for the proposed and benchmark detectors.}
    \label{fig:results}
    \vspace{-0.1cm}
\end{figure*}
\paragraph*{RL Jammer Performance}
Fig.~\ref{fig:results}(a) shows that the proposed jammer consistently attains
the highest outage probability, exceeding \(0.8\) for SNR values of \(15\) dB
and above, as the learned reward explicitly balances disruption and stealth
after accounting for Bob's defensive response. Fig.~\ref{fig:results}(b)
reports the average stealth penalty, i.e., the average value of \(g_k\)
in~\eqref{eq:reward}. The proposed jammer achieves the lowest penalty over the
entire SNR range, indicating that the inference-driven feedback enables the
jammer to avoid actions likely to trigger Bob's detector, whereas the baselines
either react myopically or ignore the detector altogether.

On the defense side,
the proposed detector is first trained offline on the \(K=4\) known jamming
classes and then adapted during deployment through~\eqref{eq:total_loss_detector}.
As a benchmark, we consider the open-set recognition method in~\cite{zhou2024},
combining OpenMax, Weibull-based score calibration, and a Siamese similarity
module on the same encoder backbone.

\paragraph*{Online Detection Performance}
Fig.~\ref{fig:results}(c) compares the detection probability of the proposed
detector with that of the benchmark in~\cite{zhou2024}. The proposed detector
maintains a detection probability between \(0.86\) and \(0.92\), versus
\(0.68\)--\(0.76\) for the benchmark, attributed to GAT-based temporal--spectral
representation learning and online clustering, which captures long-range
structure and adapts to previously unseen jamming strategies.
Fig.~\ref{fig:results}(d) shows the jammer reward remains negative against the
proposed detector for all SNR values but positive against the benchmark,
consistent with the reward definition in Section~\ref{Section_RL_jamming}:
under the proposed detector, the penalty associated with the inferred detector
response dominates the jammer's disruption gain, while the benchmark fails to
impose the same cost on the adaptive adversary.
\section{Conclusion}
We proposed a unified framework for zero-day jamming detection and adversarial benchmarking. A GAT-based online detector jointly classifies known and discovers emerging zero-day strategies without retraining, while an inference-driven RL jammer treats the receiver as a black-box and balances attack impact against stealth. Simulations show the proposed jammer synthesizes more effective and stealthier attacks than benchmarks, and the detector reliably identifies zero-day strategies that evade existing detectors.

\bibliographystyle{IEEEtran}
\bibliography{Ref2}

\end{document}